\documentclass[11pt]{article}

\usepackage[left=1in,top=1in,right=1in,bottom=1in]{geometry} 
\usepackage{booktabs}
\usepackage{threeparttable}

\makeatletter
\renewcommand{\paragraph}{%
  \@startsection{paragraph}{4}%
  {\z@}{1ex \@plus 1ex \@minus .2ex}{-1em}%
  {\normalfont\normalsize\bfseries}%
}
\makeatother

\usepackage{tikz}
\usetikzlibrary{backgrounds,positioning,fit,patterns,shadows,calc}

\usepackage{epsfig}
\usepackage{amsfonts}
\usepackage{amssymb}
\usepackage{amstext}
\usepackage{amsmath}

\usepackage{calligra}
\usepackage[T1]{fontenc}

\usepackage{amsthm, thmtools}

\usepackage{paralist}

\usepackage{nameref}
\definecolor{ForestGreen}{rgb}{0.1333,0.5451,0.1333}
\definecolor{DarkRed}{rgb}{0.8,0,0}
\definecolor{Red}{rgb}{1,0,0}
\usepackage[linktocpage=true,
	pagebackref=true,colorlinks,
	linkcolor=DarkRed,citecolor=ForestGreen,
	bookmarks,bookmarksopen,bookmarksnumbered]
	{hyperref}
\usepackage{cleveref}

\usepackage{thm-restate} 

\usepackage{xspace}
\usepackage{color}
\usepackage{enumitem}
\usepackage{comment}
\usepackage{caption}
\usepackage{subcaption}
\usepackage[noend]{algorithmic}
\usepackage[section,boxed]{algorithm}

\makeatletter
\def\thmt@refnamewithcomma #1#2#3,#4,#5\@nil{%
  \@xa\def\csname\thmt@envname #1utorefname\endcsname{#3}%
  \ifcsname #2refname\endcsname
    \csname #2refname\expandafter\endcsname\expandafter{\thmt@envname}{#3}{#4}%
  \fi
}
\makeatother

\declaretheorem[numberwithin=section,refname={Theorem,Theorems},Refname={Theorem,Theorems}]{theorem}
\declaretheorem[numberlike=theorem,refname={Lemma,Lemmas},Refname={Lemma,Lemmas}]{lemma}
\declaretheorem[numberlike=theorem,refname={Corollary,Corollaries},Refname={Corollary,Corollaries}]{corollary}

\renewcommand{\varepsilon}{\epsilon}

\newcommand{\congest}{\textsf{CONGEST} }

\newcommand{\cut}{w\xspace}

\newcommand{\id}{\ensuremath{{\sf id}}\xspace}

\def\polylog{\operatorname{polylog}}
\def\poly{\operatorname{poly}}

\newboolean{short}
\setboolean{short}{false}

\newcommand{\shortOnly}[1]{\ifthenelse{\boolean{short}}{#1}{}}
\newcommand{\onlyShort}[1]{\ifthenelse{\boolean{short}}{#1}{}}
\newcommand{\longOnly}[1]{\ifthenelse{\boolean{short}}{}{#1}}
\newcommand{\onlyLong}[1]{\ifthenelse{\boolean{short}}{}{#1}}

\onlyShort{

\setlength{\textheight}{9.2in}
\setlength{\textwidth}{6.55in}
}

\newcommand{\squishlist}{
 \begin{list}{$\bullet$}
  { \setlength{\itemsep}{0pt}
     \setlength{\parsep}{2pt}
     \setlength{\topsep}{2pt}
     \setlength{\partopsep}{0pt}
     \setlength{\leftmargin}{1.5em}
     \setlength{\labelwidth}{1em}
     \setlength{\labelsep}{0.5em} } }
\newcommand{\squishend}{
  \end{list}  }

\def\ShowComment{True}

\ifdefined\ShowComment

\def\danupon#1{\marginpar{$\leftarrow$\fbox{D}}\footnote{$\Rightarrow$~{\sf #1 --Danupon}}}
\def\hsinhao#1{\marginpar{$\leftarrow$\fbox{H}}\footnote{$\Rightarrow$~{\sf #1 --Hsin-Hao}}}

\else

\def\danupon#1{}
\def\hsinhao#1{}

\fi

\begin{document}
\begin{titlepage}
\title{Almost-Tight Distributed Minimum Cut Algorithms\thanks{The preliminary versions of this paper appeared as brief announcement papers at PODC 2014 and SPAA 2014 \cite{Nanongkai14_podc,Su14}.}} 
\author{
Danupon Nanongkai\thanks{Faculty of Computer Science, University of Vienna, W\"ahringer Stra{\ss}e 29, A-1090 Vienna, Austria. Email: {\tt danupon@gmail.com}. This work was partially done while at ICERM, Brown University USA and Nanyang Technological University, Singapore.} \\University of Vienna, Austria
\and
Hsin-Hao Su\thanks{2260 Hayward St.,~Department of EECS,~University of Michigan,~Ann Arbor,~MI 48109. Email: {\tt hsinhao@umich.edu}. This work is supported by NSF grants CCF-1217338 and CNS-1318294. This work was done while visiting MADALGO at Aarhus University, supported by Danish National Research Foundation grant DNRF84.}\\University of Michigan, USA
}

\date{}

\maketitle 

\thispagestyle{empty}
\begin{abstract}
We study the problem of computing the minimum cut in a weighted distributed message-passing networks (the \congest model). Let $\lambda$ be the minimum cut, $n$ be the number of nodes (processors) in the network, and $D$ be the network diameter. Our algorithm can compute $\lambda$ exactly in $O((\sqrt{n} \log^{*} n +D)\lambda^4 \log^2 n)$ time. To the best of our knowledge, this is the first paper that explicitly studies computing the exact minimum cut in the distributed setting. Previously, non-trivial sublinear time algorithms for this problem are known only for unweighted graphs when $\lambda\leq 3$ due to Pritchard and Thurimella's $O(D)$-time and $O(D+n^{1/2}\log^* n)$-time algorithms for computing $2$-edge-connected and $3$-edge-connected components  [ACM Transactions on Algorithms 2011].

By using the edge sampling technique of Karger [STOC 1994], we can convert this algorithm into a $(1+\epsilon)$-approximation  $O((\sqrt{n}\log^{*} n+D)\epsilon^{-5}\log^3 n )$-time algorithm for any $\epsilon>0$. This improves over the previous $(2+\epsilon)$-approximation $O((\sqrt{n}\log^{*} n+D)\epsilon^{-5}\log^2 n \log \log n)$-time algorithm  and $O(\epsilon^{-1})$-approximation $O(D + n^{\frac{1}{2}+\epsilon}\poly\log n)$-time algorithm of Ghaffari and Kuhn [DISC 2013]. Due to the lower bound of $\Omega(D+n^{1/2}/\log n)$ by Das Sarma et al. [SICOMP 2013] which holds for any approximation algorithm, this running time is {\em tight} up to a $\poly\log n$ factor.

To get the stated running time, we developed an approximation algorithm which combines the ideas of Thorup's algorithm [Combinatorica 2007] and Matula's contraction algorithm [SODA 1993]. It saves an $\epsilon^{-9} \log^{7} n$ factor as compared to applying Thorup's tree packing theorem directly. Then, we combine Kutten and Peleg's tree partitioning algorithm [J. Algorithms 1998] and Karger's dynamic programming [JACM 2000] to achieve an efficient distributed algorithm that finds the minimum cut when we are given a spanning tree that crosses the minimum cut exactly once. 
\end{abstract}




\end{titlepage}

\section{Introduction}

Minimum cut is an important measure of networks. It determines, e.g., the network vulnerability and the limits to the speed at which information can be transmitted.
%
%
While this problem has been well-studied in the centralized setting (e.g. \cite{Karger93,KS93,Karger94,Karger94b,NI92,Matula93,Gabow95,SW97,Karger00}),
 very little is known in the distributed setting, especially in the relevant context where communication links are constrained by a small {\em bandwidth} -- the so-called \congest model (cf. \Cref{sec:prelim}). 

Consider, for example, a simple variation of this problem, called {\em $\lambda$-edge-connectivity}: given an {\em unweighted} undirected graph $G$ and a {\em constant} $\lambda$, we want to determine whether $G$ is $\lambda$-edge-connected or not. In the centralized setting, this problem can be solved in $O(m+ n\lambda^2 \log n)$ time \cite{Gabow95},
thus near-linear time when $\lambda$ is a constant. (Throughout, $n$, $m$, and $D$ denotes the number of nodes, number of edges, and the network diameter, respectively.) In the distributed setting, however, non-trivial solutions exist only when $\lambda\leq 3$; this is due to algorithms of Pritchard and Thurimella \cite{PritchardT11} which can compute $2$-edge-connected and $3$-edge-connected components in $O(D)$ and $O(D+n^{1/2}\log^* n)$ time, respectively, with high probability\footnote{We say that an event holds {\em with high probability} (w.h.p.) if it holds with probability at least $1-1/n^c$, where $c$ is an arbitrarily large constant.}. This implies that the  $\lambda$-edge-connectivity problem can be solved in $O(D)$ time when $\lambda=2$ and $O(D+n^{1/2}\log^* n)$ time when $\lambda=3$. 



For the general version where input graphs could be weighted, the problem can be solved in near-linear time \cite{Karger00,Matula93,Karger94,Karger94b} in the centralized setting. In the distributed setting, the first non-trivial upper bounds are due to Ghaffari and Kuhn \cite{GhaffariK13}, who presented  $(2+\epsilon)$-approximation $O((\sqrt{n}\log^{*} n+D)\epsilon^{-5}\log^2 n \log \log n)$-time and $O(\epsilon^{-1})$-approximation $O(D + n^{\frac{1}{2}+\epsilon}\poly\log n)$-time algorithms. 
These upper bounds are complemented by a lower bound of $\Omega(D+n^{1/2}/\log n)$ for any approximation algorithm which was earlier proved by Das~Sarma~et~al.~\cite{DasSarma12} for the weighted case and later extended by \cite{GhaffariK13} to the unweighted case.  
This means that the running times of the algorithms in \cite{GhaffariK13} are tight up to a $\polylog n$ factor. Yet, it is still open whether we can achieve an approximation factor less than two in the same running time, or in fact, in any sublinear (i.e. $O(D+o(n))$) time.

\paragraph{Results.} In this paper, we present improved distributed algorithms for computing the minimum cut both exactly and approximately. Our exact deterministic algorithm for finding the minimum cut takes $O((\sqrt{n} \log^{*} n +D)\lambda^4 \log^2 n)$ time, where $\lambda$ is the value of the minimum cut.
Our approximation algorithm finds a $(1+\epsilon)$-approximate minimum cut in $O((D+\sqrt{n}\log^{*} n)\epsilon^{-5}\log^{3}n)$ time with high probability. (If we only want to compute the $(1+\epsilon)$-approximate {\em value} of the minimum cut, then the running time can be slightly reduced to $O((\sqrt{n}\log^{*} n+D)\epsilon^{-5}\log^2 n \log\log n)$.)
%
%
As noted earlier, prior to this paper there was no sublinear-time exact algorithm even when $\lambda$ is a constant greater than three, nor sublinear-time algorithm with approximation ratio less than two. 
\Cref{table:result summary} summarizes the results.

\paragraph{Techniques.} The starting point of our algorithm is Thorup's tree packing theorem \cite[Theorem~9]{Thorup07}, which shows that if we generate $\Theta(\lambda^7 \log^3 n)$ trees $T_1, T_2, \ldots$, where tree $T_i$ is the minimum spanning tree with respect to the loads induced by  $\{T_1, \ldots, T_{i-1}\}$, then one of these trees will contain exactly one edge in the minimum cut (see \Cref{sec:mincut algo} for the definition of load). 
%
%
Since we can use the $O(\sqrt{n}\log^{*} n+D)$-time algorithm of Kutten and Peleg \cite{KuttenP98} to compute the minimum spanning tree (MST), the problem of finding a minimum cut is reduced to finding the minimum cut that {\em $1$-respects a tree}; i.e., finding which edge in a given spanning tree defines a smallest cut (see the formal definition in \Cref{sec:one respect tree}). 
Solving this problem in $O(D+\sqrt{n} \log^{*}n)$ time is the first key technical contribution of this paper. 
We do this by using a simple observation of Karger \cite{Karger00} which reduces the problem to computing the sum of degree and the number of edges contained in a subtree rooted at each node. 
We use this observation along with Garay, Kutten and Peleg's {\em tree partitioning} \cite{KuttenP98,GarayKP98} to quickly compute these quantities. This requires several (elementary) steps, which we will discuss in more detail in \Cref{sec:one respect tree}.

The above result together with Thorup's tree packing theorem immediately imply that we can find a minimum cut exactly in $O((D+\sqrt{n} \log^{*}n)\lambda^{7} \log^3 n)$ time. By using Karger's random sampling result \cite{Karger94b} to bring $\lambda$ down to $O(\log n / \epsilon^2)$, we can find an $(1+\epsilon)$-approximate minimum cut in $O((D+\sqrt{n} \log^{*}n)\epsilon^{-14} \log^{10} n)$ time. 
These time bounds unfortunately depend on large factors of $\lambda$, $\log n$ and $1/\epsilon$, which make their practicality dubious. Our second key technical contribution is a new algorithm which significantly reduces these factors by combining Thorup's greedy tree packing approach with Matula's contraction algorithm \cite{Matula93}. In Matula's $(2+\epsilon)$-approximation algorithm for the minimum cut problem, he partitioned the graph into {\em components} according to the {\it spanning forest decomposition} by Nagamochi and Ibaraki \cite{NI92}. He showed that either a component induces a $(2+\epsilon)$-approximate minimum cut, or the minimum cut does not intersect with the components. In the latter case, it is safe to contract the components. Our algorithm used a similar approach, but we partitioned the graph according to Thorup's greedy tree packing approach instead of the spanning forest decomposition. We will show that either (i) a component induces a $(1+\epsilon)$-approximate minimum cut, (ii) the minimum cut does not intersect with the components, or (iii) the minimum cut 1-respect a tree in the tree packing.  This algorithm and analysis will be discussed in detail in \Cref{sec:mincut algo}.
We note that our algorithm can also be implemented in the centralized setting in $O(m + n\epsilon^{-7} \log^3 n )$ time. It is slightly worse than the current best $O(m + n\epsilon^{-3} \log^3 n )$ by Karger \cite{Karger94}. 


 
\begin{table}
\centering
\begin{tabular}{|c|c|c|}
\hline
{\bf Reference}& {\bf Time} & {\bf Approximation}\\
\hline
Pritchard\&Thurimella \cite{PritchardT11} & $O(D)$ for $\lambda\leq 2$ & exact\\
Pritchard\&Thurimella \cite{PritchardT11} & $O(\sqrt{n}\log^*n+D)$ for $\lambda\leq 3$ & exact\\
This paper & $O((\sqrt{n} \log^{*} n +D)\lambda^4 \log^2 n)$ & exact \\
\hline
Das Sarma et al. \cite{DasSarma12} & $\Omega(\frac{\sqrt{n}}{\log n} + D)$ & any\\
Ghaffari\&Kuhn \cite{GhaffariK13} & $O((\sqrt{n}\log^{*} n+D)\epsilon^{-5}\log^2 n \log \log n)$ & $2+\epsilon$\\
This paper & $O((\sqrt{n}\log^{*} n + D)\epsilon^{-5}\log^{3}n)$ & $1+\epsilon$\\
\hline
\end{tabular}
\caption{Summary of Results}\label{table:result summary}
\end{table}

\section{Preliminaries}\label{sec:prelim}

%

\paragraph{Communication Model.} We use a standard message passing network model called \congest \cite{PelegBook}. A network of processors is  modeled by an undirected unweighted  $n$-node graph $G$, where nodes model the processors and  edges model {\em $O(\log n)$-bandwidth} links between the processors. 
The processors  (henceforth, nodes) are assumed to have unique IDs in the range of $\{1, \ldots, \poly(n)\}$ and infinite computational power. We denote the ID of node $v$ by $\id(v)$. Each node has limited topological knowledge; in particular, it only knows the IDs of its neighbors and knows {\em no} other topological information (e.g., whether its neighbors are linked by an edge or not).  Additionally, we let $w:E(G)\rightarrow \{1, 2, \ldots, \poly(n)\}$ be the edge weight assignment. The weight $w(uv)$ of each edge $uv$ is known only to $u$ and $v$. As commonly done in the literature (e.g., \cite{GhaffariK13,KhanP08,LotkerPR09,KuttenP98,GarayKP98,Nanongkai13-ShortestPaths}), we will assume that the maximum weight is $\poly(n)$ so that each edge weight can be sent through an edge (link) in one round.


There are several measures to analyze the performance of distributed algorithms. One fundamental measure is the {\em running time} defined as the worst-case number of {\em rounds} of distributed communication. At the beginning of each round, all nodes wake up simultaneously. Each node $u$ then sends an arbitrary message of $B=\log n$ bits through each edge $uv$, and the message will arrive at node $v$ at the end of the round. (See \cite{PelegBook} for detail.)
%
%
The running time is analyzed in terms of number of nodes and the diameter of the network, denoted by $n$ and $D$ respectively. Since we can compute $n$ and $2$-approximate $D$ in $O(D)$ time, we will assume that every node knows $n$ and the $2$-approximate value of $D$. 


\paragraph{Minimum Cut Problem.}
Given a weighted undirected graph $G=(V,E)$, a {\em cut} $C=(S,V \setminus S)$ where $\emptyset \subsetneq S \subsetneq V$, is a partition of vertices into two non-empty sets. The {\em weight} of a cut, denoted by $w(C)$, is defined to be the sum of the edge weights crossing $C$; i.e., $w(C)=\sum_{u\in S, v\notin S} w(uv)$. Throughout the paper, we use $\lambda$ to denote the weight of the minimum cut. 
%
A $(1+\epsilon)$-approximate minimum cut is a cut $C$ whose weight $w(C)$ is such that $\lambda \leq w(C)\leq (1+\epsilon)\lambda.$ 
The (approximate) minimum cut problem is to find a cut $C=(S, V\setminus S)$ with the minimum or approximately minimum weight. In the distributed setting, this means that nodes in $S$ should output $1$ while other nodes output $0$. 



\paragraph{Graph-Theoretic Notations.}
For $G=(V,E)$, we define $V(G) = V$ and $E(G) = E$. When we analyze the correctness of our algorithms, we will always treat $G$ as an {\em unweighted multi-graph} by replacing each edge $e$ with $w(e)$ by $w(e)$ copies of $e$ with weight one. We note that this assumption is used only in the analysis, and in particular we still allow only $O(\log n)$ bits to be communicated through edge $e$ in each round of the algorithm (regardless of $w(e)$). 
For any cut $C = (S, V \setminus S)$, let $E(C)$ denote the set of edges crossing between $S$ and $V \setminus S$ in the multi-graph; thus $w(C)=|E(C)|$. Given an edge set $F \subseteq E$, we use $G/F$ to denote the graph obtained by contracting every edge in $F$. Given a partition $\mathcal{P}$ of nodes in $G$, we use $G/\mathcal{P}$ to denote the graph obtained by contracting each set in $\mathcal{P}$ into one node. Note that $E(G/\mathcal{P})$ may be viewed as the set of edges in $G$ that cross between different sets in $\mathcal{P}$. For any $U \subseteq V$, we use $G \mid U$ to denote the subgraph of $G$ induced by nodes in $U$.
For convenience, we use the subscript $*_{H}$ to denote the quantity $*$ of $H$; for example, $\lambda_{H}$ denote the value of the minimum cut of the graph $H$. A quantity without a subscript refer to the quantity of $G$, the input graph.



\section{Distributed Algorithm for Finding a Cut that 1-Respects a Tree}\label{sec:one respect tree}

In this section, we solve the following problem: Given a spanning tree $T$ on a network $G$ rooted at some node $r$, we want to find an edge in $T$ such that when we cut it, the cut defined by edges connecting the two connected component of $T$ is smallest. To be precise, for any node $v$, define $v^\downarrow$  to be the set of nodes that are descendants of $v$ in $T$, including $v$. Let $C_v=(v^\downarrow, V\setminus v^\downarrow)$. 
%
%
%
The problem is then to compute $c^* = \min_{v\in V(G)} \cut(C_v)$.
%
%
%
%
The main result of this section is the following. 

\begin{theorem}
There is an $O(D+n^{1/2}\log^*n)$-time distributed algorithm that can compute $c^*$ as well as find a node $v$ such that  $c^* = \cut(C_v)$. 
\end{theorem}

%
%
In fact, at the end of our algorithm every node $v$ knows $\cut(C_v)$. Our algorithm is inspired by the following observation used in Karger's dynamic programming \cite{Karger00}. For any node $v$, let $\delta(v)$ be the weighted degree of $v$, i.e. $\delta(v)=\sum_{u\in V(G)} w(u, v)$. Let $\rho(v)$ denote the total weight of edges whose end-points' least common ancestor in $T$ is $v$. Let $\delta^\downarrow(v)=\sum_{u\in v^\downarrow} \delta(u)$ and $\rho^\downarrow(v)=\sum_{u\in v^\downarrow} \rho(u)$. 

%
%

\begin{lemma}[Karger \cite{Karger00} (Lemma 5.9)]\label{thm:Karger}
$\cut(C_v) = \delta^\downarrow(v)-2\rho^\downarrow(v)$.
\end{lemma}

Our algorithm will make sure that every node $v$ knows $\delta^\downarrow(v)$ and $\rho^\downarrow(v)$. By \Cref{thm:Karger}, this will be sufficient for every node $v$ to compute $w(C_v)$. 
The algorithm is divided in several steps, as follows.

\newcommand{\figscale}{0.3\textwidth}
\begin{figure}
\centering
\begin{subfigure}{\figscale}
\centering
\includegraphics[page=1, clip=true, trim=7cm 9.9cm 8cm 1.1cm, width=\textwidth]{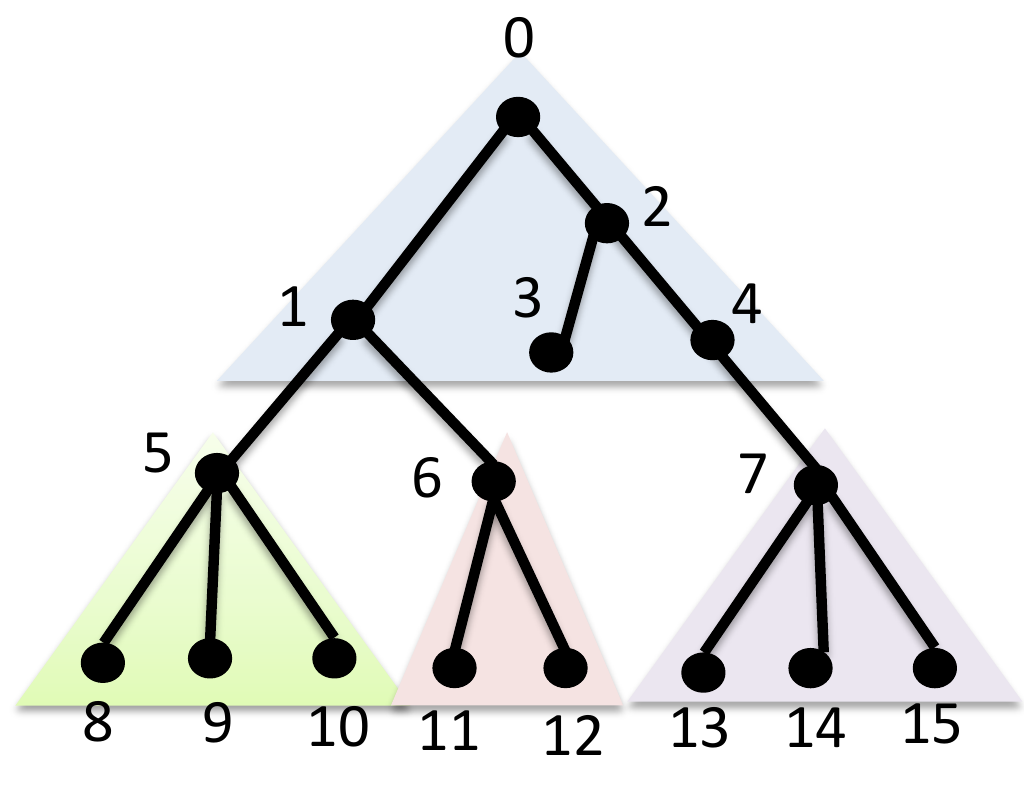}
\caption{}\label{fig:one}
\end{subfigure}
\hspace{0.01\textwidth}
\begin{subfigure}{\figscale}
\centering
\includegraphics[page=2, clip=true, trim=7cm 9.9cm 8cm 1.1cm, width=\textwidth]{mincut.pdf}
\caption{}\label{fig:two}
\end{subfigure}
\hspace{0.01\textwidth}
\begin{subfigure}{\figscale}
\centering
\includegraphics[page=3, clip=true, trim=7cm 9.9cm 8cm 1.1cm, width=\textwidth]{mincut.pdf}
\caption{}\label{fig:three}
\end{subfigure}

\vspace{0.3cm}

\begin{subfigure}{\figscale}
\centering
\includegraphics[page=4, clip=true, trim=7cm 9.9cm 8cm 1.1cm, width=\textwidth]{mincut.pdf}
\caption{}\label{fig:four}
\end{subfigure}
\hspace{0.01\textwidth}
\begin{subfigure}{\figscale}
\centering
\includegraphics[page=5, clip=true, trim=7cm 9.9cm 8cm 1.1cm, width=\textwidth]{mincut.pdf}
\caption{}\label{fig:five}
\end{subfigure}
\hspace{0.01\textwidth}
\begin{subfigure}{\figscale}
\centering
\includegraphics[page=6, clip=true, trim=7cm 9.9cm 8cm 1.1cm, width=\textwidth]{mincut.pdf}
\caption{}\label{fig:six}
\end{subfigure}
\caption{}\label{fig:example}
\end{figure}

\paragraph{Step 1: Partition $T$ into Fragments and Compute ``Fragment Tree'' $T_F$.} 
We use the algorithm of Kutten and Peleg \cite[Section 3.2]{KuttenP98} to partition nodes in tree $T$ into $O(\sqrt{n})$ subtrees, where each subtree has $O(\sqrt{n})$ diameter\footnote{To be precise, we compute a {\em $(\sqrt{n}+1, O(\sqrt{n}))$ spanning forest}. Also note that we in fact do not need this algorithm since we obtain $T$ by using Kutten and Peleg's MST algorithm, which already computes the $(\sqrt{n}+1, O(\sqrt{n}))$ spanning forest as a subroutine. See \cite{KuttenP98} for details.} 
(every node knows which edges incident to it are in the subtree containing it). This algorithm takes $O(n^{1/2}\log^*n+D)$ time. We call these subtrees {\em fragments} and denote them by $F_1, \ldots, F_k$, where $k=O(\sqrt{n})$. 
For any $i$, let $\id(F_i)=\min_{u\in F_i} \id(u)$ be the {\em ID of $F_i$}. We can assume that every node in $F_i$ knows $\id(F_i)$. This can be achieved in $O(\sqrt{n})$ time (the running time is independent of $D$) by a communication within each fragment. 
\Cref{fig:one} illustrates the tree $T$ (marked by black lines) with fragments (defined by triangular regions).

Let $T_F$ be a rooted tree obtained by contracting nodes in the same fragment into one node. This naturally defines the child-parent relationship between fragments (e.g. the fragments labeled (5), (6), and (7) in \Cref{fig:two} are children of the fragment labeled (0)). 
Let the {\em root} of any fragment $F_i$, denoted by $r_i$, be the node in $F_i$ that is nearest to the root $r$ in $T$. 
We now make every node know $T_F$: Every ``inter-fragment'' edge, i.e. every edge $(u, v)$ such that $u$ and $v$ are in different fragments, either node $u$ or $v$ broadcasts this edge and the IDs of fragments containing $u$ and $v$ to the whole network. This step takes $O(\sqrt{n}+D)$ time since there are $O(\sqrt{n})$ edges in $T$ that link between different fragments and so they can be collected by pipelining. 
Note that this process also makes every node know the roots of all fragments since, for every inter-fragment edge $(u, v)$, every node knows the child-parent relationship between two fragments that contain $u$ and $v$. 

\paragraph{Step 2: Compute Fragments in Subtrees of Ancestors.}  For any node $v$ let $F(v)$ be the set of fragments $F_i\subseteq v^\downarrow$. For any node $v$ in any fragment $F_i$, let $A(v)$ be the set of ancestors of $v$ in $T$ that are in $F_i$ or the {\em parent} fragment of $F_i$ (also let $A(v)$ contain $v$). (For example, \Cref{fig:three} shows $A(15)$.) We emphasize that $A(v)$ does not contain ancestors of $v$ in the fragments that are neither $F_i$ nor the parent of $F_i$. 
The goal of this step is to make every node $v$ knows (i) $A(v)$ and (ii) $F(u)$ for all $u\in A(v)$. 
%
%
%

First, we make every node $v$ know $F(v)$: for every fragment $F_i$ we aggregate from the leaves to the root of $F_i$ (i.e. upcast) the list of child fragments of $F_i$. This takes $O(\sqrt{n}+D)$ time since there are $O(\sqrt{n})$ fragments to aggregate and each fragment has diameter $O(\sqrt{n})$. In this process every node $v$ receives a list of child fragments of $F_i$ that are contained in $v^\downarrow$. It can then use $T_F$ to compute fragments that are descendants of these child fragments, and thus compute {\em all} fragments contained in $v^\downarrow$. 

Next, we make every node $v$ in every fragment $F_i$ know $A(v)$: every node $u$ sends a message containing its ID down the tree $T$ until this message reaches the leaves of the child fragments of $F_i$. Since each fragment has diameter $O(\sqrt{n})$ and the total number of messages sent inside each fragment is $O(\sqrt{n})$, this process takes $O(\sqrt{n})$ time (the running time is independent of $D$). 
With the following minor modifications, we can also make every node $v$ know $F(u)$ (the fragment that $u$ is in) for all $u\in A(v)$: Initially every node $u$ sends a message $(u, F')$, for every $F'\in F(u)$, to its children. Every node $u$ that receives a message $(u', F')$ from its parent sends this message further to its children {\em if $F'\notin F(u)$}. (A message $(u', F')$ that a node $u$ sends to its children should be interpreted as ``$u'$ is the lowest ancestor of $u$ such that $F'\in F(u')$''.)

\paragraph{Step 3: Compute $\delta^\downarrow(v)$.} For every fragment $F_i$, we let $\delta(F_i)=\sum_{v\in F_i} \delta(v)$ (i.e. the sum of degree of nodes in $F_i$). 
For every node $v$ in every fragment $F_i$, we will compute $\delta^\downarrow(v)$ by separately computing (i) $\sum_{u\in F_i\cap v^\downarrow} \delta(u)$ and (ii)  $\sum_{F_j\in F(v)} \delta(F_j)$. 
The first quantity can be computed in $O(\sqrt{n})$ time (regardless of $D$) by computing the sum within $F_i$ (every node $v$ sends the sum $\sum_{u\in F_i\cap v^\downarrow} \delta(u)$ to its parent). 
To compute the second quantity, it suffices to make every node know $\delta(F_i)$ for all $i$ since every node $v$ already knows $F(v)$.  To do this, we make every root $r_i$ know $\delta(F_i)$ in $O(\sqrt{n})$ time by computing the sum of degree of nodes within each $F_i$. Then, we can make every node know $\delta(F_i)$ for all $i$ by letting $r_i$ broadcast $\delta(F_i)$ to the whole network.


%

\paragraph{Step 4: Compute Merging Nodes and $T'_F$.}
We say that a node $v$ is a {\em merging node} if there are two distinct children $x$ and $y$ of $v$ such that both $x^\downarrow$ and $y^\downarrow$ contain some fragments. In other words, it is a point where two fragments ``merge''. For example, nodes $0$ and $1$ in \Cref{fig:one} are merging nodes since the subtree rooted at node $0$ (respectively node $1$) contains fragments $(5)$, $(6)$, and $(7)$ (respectively $(5)$ and $(6)$). 

%
Let $T'_F$ be the following tree: Nodes in $T'_F$ are both roots of fragments ($r_i$'s) and merging nodes. The parent of each node $v$ in $T'_F$ is its lowest ancestor in $T$ that appears in $T'_F$ (see \Cref{fig:four} for an example). Note that every merging node has at least two children in $T'_F$. This shows that there are $O(\sqrt{n})$ merging nodes. 
%
%
%
The goal of this step is to let every node know $T'_F$. 

First, note that every node $v$ can easily know whether it is a merging node or not in one round by checking, for each child $u$, whether $u^\downarrow$ contains any fragment (i.e. whether $F(u)=\emptyset$). The merging nodes then broadcast their IDs to the whole network. (This takes $O(\sqrt{n})$ time since there are $O(\sqrt{n})$ merging nodes.)
Note further that every node $v$ in $T'_F$ knows its parent in $T'_F$ because its parent in $T'_F$ is one of its ancestors in $A(v)$. 
%
So, we can make every node know $T'_F$ in $O(\sqrt{n}+D)$ rounds by letting  every node in $T'_F$ broadcast the edge between itself and its parent in $T'_F$ to the whole network.

\paragraph{Step 5: Compute $\rho^\downarrow(v)$.} We now count, for every node $v$, the number of edges whose  least common ancestors (LCA) of their end-nodes are $v$. 
For every edge $(x, y)$ in $G$, we claim that $x$ and $y$ can compute the LCA of $(x, y)$ by exchanging $O(\sqrt{n})$ messages through edge $(x, y)$. Let $z$ denote the LCA of $(x, y)$. Consider three cases (see \Cref{fig:five}). 

\smallskip\noindent
\underline{Case 1:}
First, consider when $x$ and $y$ are in the same fragment, say $F_i$. In this case we know that $z$ must be in $F_i$. Since $x$ and $y$ have the lists of their ancestors in $F_i$, they can find $z$ by exchanging these lists. There are $O(\sqrt{n})$ nodes in such list so this takes $O(\sqrt{n})$ time. 
In the next two cases we assume that $x$ and $y$ are in different fragments, say $F_i$ and $F_j$, respectively. 

\smallskip\noindent
\underline{Case 2:}
$z$ is {\em not} in $F_i$ and $F_j$. In this case, $z$ is a merging node such that $z^\downarrow$ contains $F_i$ and $F_j$. Since both $x$ and $y$ knows $T'_F$ and their ancestors in $T'_F$, they can find $z$ by exchanging the list of their ancestors in $T'_F$. There are $O(\sqrt{n})$ nodes in such list so this takes $O(\sqrt{n})$ time. 

\smallskip\noindent
\underline{Case 3:} $z$ is in $F_i$ (the case where $z$ is in $F_j$ can be handled in a similar way). In this case $z^\downarrow$ contains $F_j$. Since $x$ knows $F(x')$ for all its ancestors $x'$ in $F_i$, it can compute its lowest ancestor $x''$ such that $F(x'')$ contains $F_j$. Such ancestor is the LCA of $(x, y)$. 

\medskip\noindent
Now we compute $\rho^\downarrow(v)$ for every node $v$ by splitting edges $(x, y)$ whose LCA is $v$ into two types (see \Cref{fig:six}): (i) those that $x$ and $y$ are in different fragments from $v$, and (ii) the rest. 
For (i), note that $v$ must be a merging node. In this case one of $x$ and $y$ creates a message $\langle v\rangle$. We then count the number of messages of the form $\langle v\rangle$ for every merging node $v$ by computing the sum along the breadth-first search tree of $G$. This takes $O(\sqrt{n}+D)$ time since there are $O(\sqrt{n})$ merging nodes. 
For (ii), the node among $x$ and $y$ that is in the same fragment as $v$ creates and keeps a message $\langle v\rangle$. Now every node $v$ in every fragment $F_i$ counts the number of messages of the form $\langle v\rangle$ in $v^\downarrow\cap F_i$ by computing the sum through the tree $F_i$. Note that, to do this, every node $u$ has to send the number of messages of the form $\langle v\rangle$ to its parent, for all $v$ that is an ancestor of $u$ in the same fragment. There are $O(\sqrt{n})$ such ancestors, so we can compute the number of messages of the form $\langle v\rangle$ for every node $v$ {\em concurrently} in $O(\sqrt{n})$ time by pipelining. 


\section{Minimum Cut Algorithms}\label{sec:mincut algo}


This section is organized as follows. In \Cref{sec:thorup}, we review properties of the greedy tree packing as analyzed by Thorup \cite{Thorup07}. We use these properties to develop a $(1+\epsilon)$-approximation algorithm in \Cref{sec:meta algo}. We show how to efficiently implement this algorithm in the distributed setting in \Cref{sec:distributed implementation} and in the sequential setting in \Cref{sec:sequential algo}.

\subsection{A Review of Thorup's Work on Tree Packings}\label{sec:thorup}

In this section, we review the duality connection between the tree packing and the partition of a graph as well as their properties from Thorup's work \cite{Thorup07}.

A {\it tree packing} $\mathcal{T}$ is a multiset of spanning trees. The {\it load} of an edge $e$ with respect to $\mathcal{T}$, denoted by $\mathcal{L}^{\mathcal{T}}(e)$, is the number of trees in $\mathcal{T}$ containing $e$. Define the {\it relative load} to be $\ell^{T}(e)=\mathcal{L}^{\mathcal{T}}(e) / |\mathcal{T}|$. A tree packing $\mathcal{T} = \{T_1, \ldots, T_k \}$ is {\it greedy} if each $T_i$ is a minimum spanning tree with respect to the loads induced by $\{T_1,\ldots, T_{i-1} \}$. 

Given a tree packing $\mathcal{T}$, define its {\it packing value} $\mathtt{pack\_val}(\mathcal{T}) = 1/\max_{e \in E} \ell^{\mathcal{T}}(e)$. The packing value can be viewed as the total weight of a fractional tree packing, where each tree has weight $1/\max_{e \in E} \mathcal{L}^{\mathcal{T}}(e)$. Thus, the sum of the weight over the trees is $|T| / \max_{e \in E}\mathcal{L}^{\mathcal{T}}(e)$, which is $\mathtt{pack\_val}(\mathcal{T})$. Given a partition $\mathcal{P}$, define its {\it partition value} $\mathtt{part\_val}(\mathcal{P}) = \frac{|E(G/\mathcal{P})|}{|\mathcal{P}| - 1}$. For any tree packing $\mathcal{T}$ and partition $\mathcal{P}$, we have the weak duality:
\begin{align*}
\mathtt{pack\_val(\mathcal{T})} &= \frac{1}{\max_{e \in E}\ell^{\mathcal{T}}(e)} \\
&\leq \frac{1}{\max_{e \in E(G/\mathcal{P})}\ell^{\mathcal{T}}(e)}\\
&\leq \frac{|E(G/\mathcal{P})|}{\sum_{e \in E(G/\mathcal{P})} \ell^{\mathcal{T}}(e)} && \tag*{(since max $\geq$ avg)}\\
&\leq \frac{|E(G/\mathcal{P})|}{|\mathcal{P}| - 1} && \tag*{(since each $T \in \mathcal{T}$ contains at least $|\mathcal{P}|-1$ edges crossing $\mathcal{P}$)} \\
&= \mathtt{part\_val(\mathcal{P})}
\end{align*}

The Nash-Williams-Tutte Theorem \cite{NW61, Tutte61} states that a graph $G$ contains  $\min_{\mathcal{P}} \lfloor \frac{|E(G/\mathcal{P})|}{|\mathcal{P}| - 1} \rfloor $ edge-disjoint spanning trees. Construct the graph $G'$ by duplicating $|\mathcal{P}|-1$ edges for every edge in $G$. It follows from the Nash-Williams-Tutte Theorem that $G'$ has exactly $|E(G / \mathcal{P})|$ edge-disjoint spanning trees. By assigning each spanning tree a weight of $1/(|\mathcal{P}| - 1)$, we get a tree packing in $G$ whose packing value equals to $\frac{|E(G/\mathcal{P})|}{|\mathcal{P}| - 1}$. Therefore,
$$\max_{\mathcal{T}}\; \mathtt{pack\_val}(\mathcal{T}) = \min_{\mathcal{P}}\; \mathtt{part\_val}(\mathcal{P}).$$
We will denote this value by $\Phi$. Let $\mathcal{T}^{*}$ and $\mathcal{P}^{*}$ denote a tree packing and a partition with $\mathtt{pack\_val}(\mathcal{T}^{*}) = \Phi$ and $\mathtt{part\_val}(\mathcal{P}^{*})=\Phi$. Karger \cite{Karger00} showed the following relationship between $\Phi$ and $\lambda$ (recall that $\lambda$ is  the value of the minimum cut).
%

\begin{lemma}\label{lem:philambda} $\lambda/2 < \Phi \leq \lambda$ \end{lemma}
\begin{proof} 
$\Phi \leq \lambda$ is obvious because a minimum cut is a partition with partition value exactly $\lambda$. Consider an optimal partition $\mathcal{P}^*$. Let $C_{\min}$ be the smallest cut induced by the components in $\mathcal{P}^*$. We have
$$\lambda \leq w(C_{\min}) \leq \frac{\sum_{S \in \mathcal{P}^{*}}|E(S,V \setminus S)| }{|\mathcal{P}^{*}|} \leq \frac{2|E(G/\mathcal{P}^{*})|}{|\mathcal{P}^{*}|} < 2\Phi. \qedhere$$
\end{proof}


Thorup \cite{Thorup07} defined the {\it ideal relative loads} $\ell^{*}(e)$ on the edges of $G$ by the following.
\begin{enumerate}[noitemsep]
\item Let $\mathcal{P}^{*}$ be an optimal partition with $\mathtt{part\_val}(\mathcal{P}^*) = \Phi$. 
\item For all $e \in G / \mathcal{P}^{*}$, let $\ell^{*}(e) = 1/\Phi$.
\item For each $S \in \mathcal{P}^{*}$, recurse the procedure on the subgraph $G | S$.
\end{enumerate}

Define the following notations:
$$E^{X}_{\circ \delta} = \{e \in E \mid \ell^{X}(e) \;\circ\; \delta \}$$ 
where $X$ can be $\mathcal{T}$ or $*$, and $\circ$ can be $<$, $>$, $\leq$, $\geq$, or $=$. For example, $E^{*}_{<\delta}$ denote the set of edges with ideal relative loads smaller than $\delta$.

\begin{lemma}[\cite{Thorup07}, Lemma 14]\label{lem:phi_non_decreasing} The values of $\Phi$ are non-decreasing in the sense that for each $S \in P^{*}, \Phi_{G|S} \geq \Phi$ \end{lemma}

\begin{corollary}\label{cor:subgraphcut} Let $0 \leq l \leq 1/\Phi$. Each component $H$ of the graph $(V, E^{*}_{\leq l})$ must have edge-connectivity of at least $\Phi$. \end{corollary} 
\begin{proof}
Accroding to how the ideal relative load was defined and Lemma \ref{lem:phi_non_decreasing}, we must have $\Phi_H \geq \Phi$. By Lemma \ref{lem:philambda}, $\lambda_{H} \geq \Phi_{H} \geq \Phi$.
\end{proof}

Thorup showed that the relative loads of a greedy tree packing with a sufficient number of trees approximate the ideal relative loads, due to the fact that greedily packing the trees simulates the multiplicative weight update method. He showed the following lemma.
\begin{lemma}[\cite{Thorup07}, Proposition 16]\label{lem:approxload} A greedy tree packing $\mathcal{T}$ with at least $(6 \lambda \ln m) / \epsilon^2$ trees, $\epsilon < 2$ has $|\ell^{\mathcal{T}}(e) - \ell^{*}(e)| \leq \epsilon /\lambda$ for all $e \in E$.\end{lemma}

\subsection{Algorithms}\label{sec:meta algo}


In this section, we show how to approximate the value of the minimum cut as well as how to find an approximate minimum cut.
%
%

\paragraph{Algorithm for computing minimum cut value.} 
The main idea is that if we have a nearly optimal tree packing, then either $\lambda$ is close to $2\Phi$ or all the minimum cuts are crossed exactly once by some trees in the tree packing.

\begin{lemma}\label{lem:approxvalue} Suppose that $\mathcal{T}$ is a greedy tree packing with at least $6 \lambda \ln m / \epsilon^2$ trees, then $\lambda \leq (2+\epsilon)\cdot \mathtt{pack\_val}(\mathcal{T})$. Furthermore, if there is a minimum cut $C$ such that it is crossed at least twice by every tree in $\mathcal{T}$, then $(2+\epsilon)\cdot \mathtt{pack\_val}(\mathcal{T}) \leq (1+\epsilon/2)\lambda$. \end{lemma}
\begin{proof}By \Cref{lem:approxload,lem:philambda}, $1/{\mathtt{pack\_val}(\mathcal{T})} \leq 1/{\mathtt{pack\_val}(\mathcal{T}^*)} + \epsilon / \lambda  \leq 2/\lambda + \epsilon /\lambda$. Therefore, $\lambda \leq (2+\epsilon)\cdot \mathtt{pack\_val}(\mathcal{T})$.

If each tree in $\mathcal{T}$ crosses $C$ at least twice, we have $\sum_{e \in C} \ell^{\mathcal{T}}(e) \geq 2$. Therefore, 
\begin{equation}\label{eqn:packval} 2/\lambda \leq \sum_{e \in C} \ell^{\mathcal{T}}(e) / w(C) \leq \max_{e\in C}\ell^{\mathcal{T}}(e)  \leq 1/\mathtt{pack\_val}(\mathcal{T})\,.\end{equation} 
This implies that $(2+\epsilon)\cdot \mathtt{pack\_val}(\mathcal{T}) \leq (1+\epsilon/2)\lambda$.
\end{proof}

Using \Cref{lem:approxvalue}, we can obtain a simple algorithm for $(1+\epsilon)$-approximating the minimum cut {\em value}. First, greedily pack $\Theta(\lambda \log n / \epsilon^2)$ trees and compute the minimum cut that 1-respects the trees (using our algorithm in \Cref{sec:one respect tree}). Then, output the smaller value between the minimum cut found and $(2+\epsilon)\cdot \mathtt{pack\_val}(\mathcal{T})$. The running time is discussed in Section \ref{sec:distributed implementation}. 

\paragraph{Algorithm for finding a minimum cut.} 
More work is needed to be done if we want to {\em find} the $(1+\epsilon)$-approximate minimum cut (i.e.~each node wants to know which side of the cut it is on). Let $\epsilon' = \Theta(\epsilon)$ be such that $(1-2\epsilon')\cdot(1-\epsilon') = 1/(1 + \epsilon)$. Let $l_{a} = (1-2\epsilon')/\mathtt{pack\_val(\mathcal{T})}$. We describe our algorithm in \Cref{alg:approxmincut}.
%
%
%
\begin{algorithm}[H]
\caption{\textsc{Approx-Min-Cut($G$)}}
\begin{algorithmic}[1]\label{alg:approxmincut}

\STATE \label{line:treepacking}Find a greedy tree packing $\mathcal{T}$ with $(6\lambda \ln m) / \epsilon'^2$ trees in $G$.

\STATE Let $C^*$ be the minimum cut among cuts that $1$-respect a tree in ${\cal T}$.

\STATE Let $l_a = (1-2\epsilon')/\mathtt{pack\_val(\mathcal{T})}$.
\IF{$(V, E^{\mathcal{T}}_{< l_a})$ has more than $(1-\epsilon')|V|$ components}
	\STATE Let $C_{\min}$ be the smallest cut induced by the components in $(V, E^{\mathcal{T}}_{< l_a})$.
\ELSE
	\STATE Let $C_{\min}$ be the cut returned by \textsc{Approx-Min-Cut($G / E^{T}_{< l_a}$)}.
\ENDIF
\STATE Return the smaller cut between $C^{*}$ and $C_{\min}$.
\end{algorithmic}
\end{algorithm}
%
%
The main result of this subsection is the following theorem.

\begin{theorem}\label{thm:approximate mincut algo}
\Cref{alg:approxmincut} gives a $(1+\epsilon)$-approximate minimum cut.
\end{theorem}

The rest of this subsection is devoted to proving \Cref{thm:approximate mincut algo}. First, observe that if a minimum cut is crossed exactly once by a tree in $\mathcal{T}$, then $C^{*}$ must be a minimum cut. Otherwise, $C$ is crossed at least twice by every tree in $\mathcal{T}$. In this case, we will show that the edges of every minimum cut will be included in $E^{\mathcal{T}}_{\geq l_{a}}$. As a result, we can contract each connected component in the partition $(V,E^{\mathcal{T}}_{< l_{a}})$ without contracting any edges of the minimum cuts.

If $(V,E^{\mathcal{T}}_{< l_{a}})$ has at most $(1-\epsilon')|V|$ components, then we contract each component and then recurse. The recursion can only happen at most $O(\log n / \epsilon)$ times, since the number of nodes reduces by a $(1-\epsilon')$ factor in each level. On the other hand, if $(V,E^{\mathcal{T}}_{< l_{a}})$ has more than $(1-\epsilon')|V|$ components, then we will show that one of the components induces an approximate minimum cut.

\begin{lemma}Let $C$ be a minimum cut such that $C$ is crossed at least twice by every tree in $\mathcal{T}$. For all $e \in C$, $\ell^{\mathcal{T}}(e) \geq (1-2\epsilon')/\mathtt{pack\_val(\mathcal{T})}$.\end{lemma}
\begin{proof}
The idea is to show that if an edge in $E(C)$ has a small relative load, then the average relative load over the edges in $E(C)$ will also be small. However, since each tree cross $E(C)$ twice, the average relative load should not be too small. Otherwise, a contradiction will occur.

Let $l_0 = \min_{e \in C} \ell^{*}(e)$ be the minimum ideal relative load over the edges in $E(C)$. Consider the induced subgraph $(V,E^{*}_{\leq l_0})$. $E(C)$ must contain some edges in a component of $(V,E^{*}_{\leq l_0})$, say component $H$. Notice that two endpoints of an edge in a minimum cut must lie on different sides of the cut. Therefore, $C \cap H$ must be a cut of $H$. By Corollary \ref{cor:subgraphcut}, $w(C \cap H) \geq \Phi$. Therefore, more than $\Phi$ edges in $C$ have ideal relative loads equal to $l_0$. Since the maximum relative load of an edge is at most $\frac{1}{\Phi}$, $\sum_{e \in C} \ell^{\mathcal{T}^*}(e) \leq \Phi\cdot l_0 + (\lambda - \Phi)\cdot \frac{1}{\Phi} = \Phi \cdot l_0 + \frac{\lambda}{\Phi} - 1 < \Phi\cdot l_0 + 1$, where the last inequality follows by Lemma \ref{lem:philambda} that $\lambda< 2\Phi$.

On the other hand, since each tree in $\mathcal{T}$ crosses $C$ at least twice, $\sum_{e \in C} \ell^{\mathcal{T}}(e) \geq 2$. By Lemma \ref{lem:approxload}, $\sum_{e \in C} \ell^{{*}}(e) \geq 2 - \epsilon'$.
Therefore, $\Phi \cdot l_0 +1>2 - \epsilon'$, which implies 
\begin{align*} l_0 &\geq (1-\epsilon') \cdot \frac{1}{\Phi} > \frac{1}{\Phi} - \frac{2\epsilon'}{\lambda} && \tag*{$\lambda < 2\Phi$} \\
&\geq 1/\mathtt{pack\_val}(\mathcal{T}) - \frac{3\epsilon'}{\lambda} && \tag*{By Lemma \ref{lem:approxload}}
\end{align*}
Therefore, by Lemma \ref{lem:approxload} again, for any $e \in E(C)$, $\ell^{\mathcal{T}}(e) \geq l_0 - \epsilon' / \lambda > 1/\mathtt{pack\_val(\mathcal{T})} - 4\epsilon'/ \lambda \geq (1-2\epsilon')/\mathtt{pack\_val(\mathcal{T})}$, where the last inequality follows from equation (\ref{eqn:packval}).
\end{proof}

\begin{lemma} Let $C_{\min}$ be the smallest cut induced by the components in $(V, E^{\mathcal{T}}_{< l_a})$. If $(V,E^{\mathcal{T}}_{<l_a})$ contains at least $(1-\epsilon')|V|$ components, then $w(C_{\min}) \leq (1+\epsilon)\lambda$. \end{lemma}
\begin{proof} Let $\mathop{comp}(V,E^{\mathcal{T}}_{<l_a})$ denote the collection of connected components in $(V,E^{\mathcal{T}}_{<l_a})$, and $n'$, the number of connected components in $(V,E^{\mathcal{T}}_{<l_a})$.
By an averaging argument, we have
\begin{equation}\label{eqn:cmin}
w(C_{\min}) \leq \frac{\sum_{S \in \mathop{comp}(V,E^{\mathcal{T}}_{<l_a}) }|E(S, V \setminus S)|}{n'} = \frac{2|E(G / E^{\mathcal{T}}_{<l_a})|}{n'} \leq \frac{2|E(G / E^{\mathcal{T}}_{<l_a})|}{(1-\epsilon')\cdot |V|}
\end{equation}
Next we will bound $|E(G / E^{\mathcal{T}}_{<l_a})|$. Note that for each $e \in E(G / E^{\mathcal{T}}_{<l_a})$, $\ell^{\mathcal{T}}(e) \geq (1-2\epsilon')/\mathtt{pack\_val}(\mathcal{T})$.
\begin{align}
\sum_{e \in E(G / E^{\mathcal{T}}_{<l_a})} \ell^{\mathcal{T}}(e) &\geq |E(G / E^{\mathcal{T}}_{<l_a})| \cdot (1-2\epsilon')\cdot \left(\frac{1}{\mathtt{pack\_val}(\mathcal{T})} \right) \nonumber \\
&\geq |E(G / E^{\mathcal{T}}_{<l_a})| \cdot \left(1- {2\epsilon'} \right) \cdot \frac{2}{\lambda}\,. && \mbox{(by \Cref{eqn:packval})}\label{eqn:sum lower bound}
\end{align}
On the other hand, 
\begin{align}
\sum_{e \in E(G / E^{\mathcal{T}}_{<l_a})} \ell^{T}(e) &\leq |V| - 1, \label{eqn:sum upper bound}
\end{align} 
since each tree in $\mathcal{T}$ contains $|V| - 1$ edges. 
\Cref{eqn:sum lower bound,eqn:sum upper bound} together imply that 
$$|E(G/E^{\mathcal{T}}_{< l_a})|\leq \frac{\lambda \cdot |V|}{2(1-2\epsilon')}.$$
By plugging in this into (\Cref{eqn:cmin}), we get that
$$w(C_{\min}) \leq \frac{\lambda}{(1-2\epsilon')(1-\epsilon')} \leq (1+\epsilon)\lambda\,. \qedhere$$\end{proof}

\subsection{Distributed Implementation}\label{sec:distributed implementation}

In this section, we describe how to implement Algorithm \ref{alg:approxmincut} in the distributed setting. To compute the tree packing $\mathcal{T}$, it is straightforward to apply $|\mathcal{T}|$ minimum spanning tree computations with edge weights equal to their current loads. This can be done in $O(|\mathcal{T}|(D+ \sqrt{n} \log^{*} n))$ rounds by using the algorithm of Kutten and Peleg \cite{KuttenP98}. 

We already described how to computes the minimum cut that 1-respects a tree in $O(D+\sqrt{n} \log^{*} n)$ rounds in \Cref{sec:one respect tree}. To compute $l_a$, it suffices to compute $\mathtt{pack\_val}(\mathcal{T})$. To do this, each node first computes the largest relative load among the edges incident to it. By using the upcast and downcast techniques, the maximum relative load over all edges can be aggregated and boardcast to every node in $O(D)$ time. Therefore, we can assume that every node knows $l_a$ now. Now we have to determine whether $(V, E^{\mathcal{T}}_{<l_a})$ has more than $(1-\epsilon')|V|$ components or not. This can be done by first removing the edges incident to each node with relative load at least $l_a$. Then label each node with the smallest ID of its reachable nodes by using Thurimella's connected component identification algorithm \cite{Thurimella97} in $O(D + \sqrt{n} \log^{*} n)$ rounds. The number of nodes whose label equals to its ID is exactly the number of connected component of the subgraph. This number can be aggregated along the BFS tree in $O(D)$ rounds after every node is labeled.

If $(V, E^{\mathcal{T}}_{<l_a})$ has more than $(1-\epsilon')|V|$ components, then we will compute the cut values induced by each component of $(V, E^{\mathcal{T}}_{<l_a})$. We show that it can be done in $O(D+\sqrt{n})$ rounds in \Cref{sec:componentcut}. On the contrary, if $(V, E^{\mathcal{T}}_{<l_a})$ has less than $(1-\epsilon')|V|$ components, then we will contract the edges with load less than $l_a$ and then recurse. The contraction can be easily implemented by setting the weights of the edges inside contracted components to be $-1$, which is strictly less than the load of any edges. The MST computation will automatically treat them as contracted edges, since an MST must contain exactly $n'-1$ edges with weights larger than $-1$, where $n'$ is the number of connected components. \footnote{We note that the MST algorithm of \cite{KuttenP98} allows negative-weight edges.} 

%

\paragraph{Time analysis.}
Suppose that we have packed $t$ spanning trees throughout the entire algorithm, the running time will be $O(t(D+ \sqrt{n} \log^{*} n))$. Note that $t = O(\epsilon^{-3} \lambda \log^2 n)$, because we pack at most $O(\epsilon^{-2}\lambda \log n )$ spanning trees in each level of the recursion and there can be at most $O(\epsilon^{-1} \log n)$ levels, since the number of nodes reduces by a $(1-\epsilon')$ factor in each level.
%
%
The total running time is $O(\epsilon^{-3} \lambda \log^2 n \cdot (D+ \sqrt{n} \log^{*} n))$.

\paragraph{Dealing with graphs with high edge connectivity.} For graphs with $\lambda = \omega(\epsilon^{-2} \log n)$, we can use the well-known sampling result from Karger's \cite{Karger94b} to construct a subgraph $H$ that perserves the values of all the cuts within a $(1\pm\epsilon)$ factor (up to a scaling) and has $\lambda_{H} = O(\epsilon^{-2} \log n)$. Then we run our algorithm on $H$.

\begin{lemma}[\cite{Karger94}, Corollary 2.4]\label{lem:karger}
Let $G$ be any graph with minimum cut $\lambda$ and let $p=2(d+2)(\ln n)/(\epsilon^2 \lambda)$. Let $G(p)$ be a subgraph of $G$ with the same vertex set, obtained by including each edge of $G$ with probability $p$ independently. Then the probability that the value of some cut in $G(p)$ has value more than $(1+\epsilon)$ or less than $(1-\epsilon)$ times its expected value is $O(1/n^d)$. 
\end{lemma}

%
%
In particular, let $\epsilon' = \Theta(\epsilon)$ such that $(1+\epsilon) = (1+\epsilon')^2/(1-\epsilon')$. First we will compute $\lambda'$, a 3-approximation of $\lambda$, by using Ghaffari and Kuhn's algorithm. Let $p = 6(d+2) \ln n/(\epsilon'^2\lambda')$ and $H = G(p)$. Since $p$ is at least $2(d+2)  \ln n/(\epsilon'^2\lambda)$, by Lemma \ref{lem:karger}, for any cut $C$, w.h.p.~$(1 - \epsilon')p \cdot w_{G}(C) \leq w_{H_{i}}(C) \leq (1 + \epsilon')p \cdot w_{G}(C)$.
Let $C^{*}$ be the $(1+\epsilon')$-approximate minimum cut we found in $H$. We have that w.h.p.~for any other cut $C'$, $$w_{G}(C^{*}) \leq \frac{1}{p} \cdot \frac{w_{H_i}(C^{*})}{1-\epsilon'} \leq \frac{1}{p} \cdot \frac{(1+\epsilon')\lambda_{H}}{1-\epsilon'} \leq \frac{1}{p} \cdot \frac{(1+\epsilon')w_{H_{i}}(C')}{1-\epsilon'} \leq \frac{(1+\epsilon')^2}{1-\epsilon'} \cdot w_{G}(C') = (1+\epsilon) w_{G}(C')$$

Thus, we will find an $(1+\epsilon)$-approximate minimum cut in $O(\epsilon^{-5}\log^3 n(D+\sqrt{n} \log^{*} n))$ rounds.

\paragraph{Computing the exact minimum cut.} To find the exact minimum cut, first we will compute a 3-approximation of $\lambda$, $\lambda'$, by using Ghaffari and Kuhn's algorithm \cite{GhaffariK13} in $O(\lambda \log n \log \log n (D+\sqrt{n} \log^{*} n))$ rounds.\footnote{Ghaffari and Kuhn's result runs in $O(\log^2 n \log \log n (D+\sqrt{n} \log^{*}n))$ rounds. However, without using Karger's random sampling beforehand, it runs in $O(\lambda \log n \log \log n (D+\sqrt{n} \log^{*} n))$ rounds, which will be absorbed by the running time of our algorithm for the exact minimum cut.} Now since $\lambda \leq \lambda' \leq 3\lambda$, by applying our algorithm with $\epsilon = 1/(\lambda'+1)$, we can compute the exact minimum cut in $O(\lambda^4 \log^2 n (D+\sqrt{n} \log^{*} n))$ rounds.

\paragraph{Estimating the value of $\lambda$.} As described in \Cref{sec:meta algo}, we can avoid the recursion if we just want to compute an approximation of $\lambda$ without actually finding the cut. This gives an algorithm that runs in $O(\epsilon^{-2} \lambda \log n \cdot (D+ \sqrt{n} \log^{*} n))$ time. Also, the exact value of $\lambda$ can be computed in $O((\lambda^3 + \lambda \log\log n) \log n(D+\sqrt{n} \log^{*} n))$ rounds. Notice that the $\lambda \log \log n$ factor comes from Ghaffari and Kuhn's algorithm for approximating $\lambda$ within a constant factor. Similarly, using Karger's sampling result, we can $(1+\epsilon)$-approximate the value of $\lambda$ in $O(\epsilon^{-5} \log^2 n \log \log n (D+\sqrt{n} \log^{*} n))$ rounds.

\subsection{Sequential Implementation}\label{sec:sequential algo}

We show that Algorithm \ref{alg:approxmincut} can be implemented in the sequential setting in $O(\epsilon^{-3} \lambda (m + n \log n) \log n )$ time. To get the stated bound, we will show that the number of edges decreases geometrically each time we contract the graph. 

\begin{lemma}
If $(V, E^{\mathcal{T}}_{< l_a})$ has less than $(1-\epsilon')|V|$ components, then $|E(G / E^{\mathcal{T}}_{< l_a})| \leq |E(G)|/(1+\epsilon')$.
\end{lemma}
\begin{proof}
Consider a component $S$ of $(V, E^{\mathcal{T}}_{< l_a})$. Since $E(S) \subseteq E^{\mathcal{T}}_{<l_a}$ and $|T \cap E(S)| \geq |S| - 1$, we have $|S| - 1 \leq \sum_{e \in S} \ell^{\mathcal{T}}(e) < l_a |E(S)|$. By summing this inequality over all components of $(V, E^{\mathcal{T}}_{< l_a})$, we have \begin{equation}\label{eqn:load1} l_a |E^{\mathcal{T}}_{< l_a}|\geq |V| - |V(G / E^{\mathcal{T}}_{< l_a}) | > |V| - (1-\epsilon')|V| = \epsilon' |V| \end{equation}

If we sum up the relative load over each $e \in E(G/E^{\mathcal{T}}_{< l_a})$, we have
\begin{equation}\label{eqn:load2}
l_{a} |E(G/E^{\mathcal{T}}_{< l_a})| \leq \sum_{e \in E(G/E^{\mathcal{T}}_{< l_a})} \ell^{\mathcal{T}}(e) \leq |V|
\end{equation}

Dividing (\ref{eqn:load1}) by (\ref{eqn:load2}), we have $|E^{\mathcal{T}}_{< l_a}|/|E(G/E^{\mathcal{T}}_{< l_a})| > \epsilon'$ and therefore, $|E(G/E^{\mathcal{T}}_{< l_a})| < (|E^{\mathcal{T}}_{< l_a}| + |E(G/E^{\mathcal{T}}_{< l_a})|) / (1+\epsilon') = |E(G)| / (1+\epsilon')$.
\end{proof}

Let $\mathtt{MST}(n, m)$ denote the time needed to find an MST in a graph with $n$-vertices and $m$-edges. Note that Karger \cite{Karger00} showed that the values of the cuts that 1-respect a tree can be computed in linear time. The total running time of Algorithm \ref{alg:approxmincut} will be
$$O\left(\epsilon'^{-2} \lambda \log n \cdot \sum_{i=0}^{\infty} \mathtt{MST}(n(1-\epsilon')^i,m/(1+\epsilon')^i)\right).$$
We know that $\mathtt{MST}(n,m) = O(m)$ by using the randomized linear time algorithm from \cite{KKT95} and notice that $\epsilon = \Theta(\epsilon')$, the running time will be at most $O(\epsilon^{-3} \lambda m \log n )$.

If the graph is dense or the cut value is large, we may want to use the sparsification results to reduce $m$ or $\lambda$. First estimate $\lambda$ up to a factor of 3 by using Matula's algorithm \cite{Matula93} that runs in linear time. By using Nagamochi and Ibaraki's sparse certificate algorithm \cite{NI92}, we can get the number of edges down to $O(n\lambda)$. By using Karger's sampling result, we can bring $\lambda$ down to $O(\log n / \epsilon^2)$. The total running time is therefore $O(m + \epsilon^{-7} n \log^3 n )$ (by plugging $\lambda=\log n / \epsilon^2$ and $m=n\log n / \epsilon^2$ in the running time in the previous paragraph).  \footnote{In this case, we can also use Prim's deterministic MST algorithm without increasing the total running time. This is because Prim's algorithm runs in $O(m + n\log n)$ time, the $n \log n$ term will be absorbed by $m$, as we have used $m=n\log n / \epsilon^2$.}


%
%
%

\paragraph{Acknowledgment:} D. Nanongkai would like to thank Thatchaphol Saranurak for bringing Thorup's tree packing theorem \cite{Thorup07} to his attention.


  \let\oldthebibliography=\thebibliography
  \let\endoldthebibliography=\endthebibliography
  \renewenvironment{thebibliography}[1]{%
    \begin{oldthebibliography}{#1}%
      \setlength{\parskip}{0ex}%
      \setlength{\itemsep}{0ex}%
  }%
  {%
    \end{oldthebibliography}%
  }
{ 
  \onlyShort{  \newpage}
\bibliographystyle{plain}
\bibliography{references}
}

\appendix
\section*{Appendix}

\section{Finding cuts with respect to connected components}\label{sec:componentcut}

In this section, we solve the following problem. We are given a set of connected components $\{H_1, H_2, \ldots, H_k\}$ of the network $G$ (each node knows which of its neighbors are in the same connected component), and we want to compute, for each $i$, the value $w(C_i)$ where $C_i$ is the cut with respect to $H_i$; i.e., $C_i=(V(H_i), V(G)\setminus V(H_i))$. Every node in $C_i$ should know $w(C_i)$ in the end. We show that this can be done in $O(n^{1/2}+D)$ rounds. 
The main idea is to deal with ``big'' and ``small'' components separately, where a component is big if it contains at least $n^{1/2}$ nodes and it is small otherwise. There are at most $n^{1/2}$ big components, and thus the cut value information for these components can be aggregated quickly through the BFS tree of the network. The cut value of each small component will be computed locally within the component. The detail is as follows. 

First, we determine for each component $H_i$ whether it is big or small, which can be done by simply counting the number of nodes in each component, such as the following. Initially, every node sends its ID to its neighbors in the same component. Then, for $n^{1/2}+1$ rounds, every node sends the smallest ID it has received so far to its neighbors in the same component. For each node $v$, let $s_v$ be the smallest ID that $v$ has received after $n^{1/2}+1$ rounds. If $s_v$ is $v$'s own ID, it construct a BFS tree $T_v$ of depth at most $n^{1/2}+1$, and use $T_v$ to count the number of nodes in $T_v$. (There will be no congestion caused by this algorithm since no other node within distance $n^{1/2}+1$ from $v$ will trigger another BFS tree construction.)  If the number of nodes in $T_v$ is at most $n^{1/2}$, then $v$ broadcasts to the whole network that the component containing it is small.


%

Now, to compute $w(C_i)$ for a small component $H_i$, we simply construct a BFS tree rooted at the node with smallest ID in $C_i$ and compute the sum $\sum_{u\in V(H_i), v\notin V(H_i)} w(u, v)$ through this tree. To compute $w(C_i)$ for a big component $H_j$, we compute the sum $\sum_{u\in V(H_i), v\notin V(H_i)} w(u, v)$ thorough the BFS tree of network $G$. Since there are at most $n^{1/2}$ big components, this takes $O(n^{1/2}+D)$ time. 


%
%

\end{document}